\documentclass{ws-ijfcs}
\usepackage{cite}
\usepackage{enumerate}
\usepackage{url}
\usepackage{complexity}
\usepackage{amssymb}
\usepackage{endnotes}
\usepackage{listings}
\usepackage{mathtools}
\usepackage{verbatim}
\usepackage{mathrsfs}
\usepackage{shuffle}
\usepackage{longtable}
\usepackage{enumitem}
\usepackage{etoolbox}
\usepackage{float}
\usepackage{color}
\usepackage{wasysym}

\DeclarePairedDelimiter{\ceil}{\lceil}{\rceil}

\newclass{\DCM}{DCM}
\newclass{\DCMNE}{DCM_{NE}}
\newclass{\TwoDCM}{2DCM}
\newclass{\NCM}{NCM}
\newclass{\DPCM}{DPCM}
\newclass{\NPCM}{NPCM}
\newclass{\TRE}{TRE}

\newclass{\DPDA}{DPDA}
\newclass{\NPDA}{NPDA}
\newclass{\NCA}{NCA}
\newclass{\DCA}{DCA}

\newcommand\set[1]{\left\{#1\right\}}

\newcommand{\sst}{\ensuremath{\mid}}

\DeclareMathOperator{\pref}{pref}
\DeclareMathOperator{\suff}{suff}
\DeclareMathOperator{\infx}{inf}
\DeclareMathOperator{\outf}{outf}

\newcounter{theo}

\newcounter{proposition}
\newcounter{corollary}
\newcounter{definition}
\newcounter{example}

\newtheorem{prop}[proposition]{Proposition}
\newtheorem{cor}[corollary]{Corollary}
\newtheorem{defin}[definition]{Definition}
\newtheorem{exam}[example]{Example}

\newtheorem{claim}{Claim}

\begin{document}

\markboth{J. Eremondi, O. H. Ibarra, I. McQuillan}
{On the Density of Context-Free and Counter Languages}

%
\catchline{}{}{}{}{}
%

\title{On the Density of Context-Free and Counter Languages\thanks{Electronic version of an article published as [International Journal of Foundations of Computer Science, 29, 02, 2018, 233--250] [10.1142/S0129054118400051] \copyright\ [copyright World Scientific Publishing Company] [https://www.worldscientific.com/worldscinet/ijfcs]}}

\author{Joey Eremondi}

\address{Department of Computer Science\\ University of British Columbia, Vancouver, BC V6T 1Z4, Canada\\
\email{joey.eremondi@alumni.ubc.ca}}

\author{Oscar H. Ibarra\footnote{Supported, in part, by
NSF Grant CCF-1117708.}}

\address{Department of Computer Science\\ University of California, Santa Barbara, CA 93106, USA\\
\email{ibarra@cs.ucsb.edu}
}

\author{Ian McQuillan\footnote{Supported, in part, by a grant from the Natural Sciences and Engineering Research Council of Canada.}}

\address{Department of Computer Science, University of Saskatchewan\\
	Saskatoon, SK S7N 5A9, Canada\\
\email{mcquillan@cs.usask.ca}
}

\maketitle

\begin{history}
\received{(Day Month Year)}
\accepted{(Day Month Year)}
\comby{(xxxxxxxxxx)}
\end{history}

\begin{abstract}
A language $L$ is said to be {\em dense} if every word in the universe is an 
infix of some word in $L$. This notion has been generalized from the infix operation to arbitrary word operations $\varrho$ in place of the infix operation
($\varrho$-dense, with infix-dense being the standard notion of dense). It is shown here that it is decidable, for a language $L$ accepted by a one-way nondeterministic reversal-bounded pushdown automaton, whether $L$ is infix-dense. However, it becomes
undecidable for both deterministic pushdown automata (with no reversal-bound),
and for nondeterministic one-counter automata. When examining suffix-density,
it is undecidable for more restricted families such as deterministic one-counter automata that make three reversals on the counter, but it is decidable with less reversals.
Other decidability results are also presented on dense languages, and contrasted with a marked version called $\varrho$-marked-density. Also, new languages are demonstrated to be outside various deterministic language families 
after applying different deletion operations from smaller families.
Lastly, bounded-dense languages are defined and examined.
\end{abstract}


\section{Introduction}

A language $L\subseteq \Sigma^*$ is dense if the set of all infixes of $L$ is equal to $\Sigma^*$ \cite{TheoryOfCodes}. This notion is relevant to the
theory of codes. Indeed, a language being dense is connected with the notions of independent sets \cite{CodesHandbook}, maximal independent sets, codes \cite{JKT}, and
disjunctive languages \cite{denseIto,shyr}.

Dense languages have been studied in \cite{denseIto,shyr} and generalized
from density to $\varrho$-density \cite{JKT}, 
where $\varrho$ is an arbitrary word
operation used in place of the infix-operation in the definition.
Some common examples are prefix-dense (coinciding with left dense in
\cite{denseIto}), suffix dense (coinciding with right dense in \cite{denseIto}),
infix dense (usual notion of density), outfix dense, embedding dense, and others from \cite{JKT}. Each type connects with a generalized notion of independent sets and codes.

It has long been known that universality of a language $L$ 
(is $L = \Sigma^*$?) is 
undecidable for $L$  accepted by a one-way nondeterministic one-counter automaton whose counter makes only one reversal, i.e., \ in an accepting computation, after decreasing the counter, it
can no longer increase again \cite{Baker1974}. This shows immediately that with the identity operation, it is undecidable if $L$ in this family is identity-dense. 
In contrast, the universality problem is known to be decidable for
one-way deterministic reversal-bounded multicounter 
languages \cite{Ibarra1978}, but these languages are not closed under taking suffix, infix, or outfix \cite{EIMDeletion2015}.
However, to decide the property of infix-density, in this paper we can show
contrasting results.
\begin{enumerate}
\item Infix-density is decidable for $L$ accepted by a nondeterministic pushdown automaton where the pushdown is reversal-bounded (there is at most a fixed number of switches between increasing and decreasing the size of the pushdown).
\item Infix-density is undecidable for $L$ accepted by a nondeterministic one-counter automaton (with no reversal-bound).
\item Infix-density is undecidable for $L$ accepted by a deterministic pushdown automaton (with no reversal-bound).
\end{enumerate}
Thus, it is surprisingly possible to decide if the set of all infixes of a nondeterministic reversal-bounded pushdown automaton gives universality, when it is undecidable with the identity operator for much smaller families. 

Furthermore, if the question is altered to change the type of density from infix-density to either suffix-density or prefix-density, then it is undecidable even for nondeterministic one-counter automata that make one counter reversal 
(coinciding with the result for identity-density). 
Suffix-density is decidable however for deterministic one-counter automata that makes one counter reversal, but is undecidable when there is either two more reversals, or two counters that both make one reversal. Thus suffix-density is often impossible to decide when infix-density is decidable. Prefix density is decidable for all deterministic reversal-bounded multicounter languages.

Contrasts are made between deciding if applying an operation $\varrho$
to a language gives $\Sigma^*$ and deciding if 
$\$ \Sigma^* \$$ (with $\$ \notin\Sigma$) is a subset of $\varrho$ applied
 to $L \subseteq (\Sigma \cup \{\$\})^*$. If this condition, 
$\$\Sigma^*\$ \subseteq \varrho(L)$, is true, the language is said to be $\varrho$-marked-dense.
In contrast to infix-density, infix-marked-density is undecidable with
 only one-way deterministic one-counter 3-reversal-bounded languages,
 and for the outfix operation with many families as well.
 Results are summarized in Table \ref{tab:summary}.

 In addition, new languages $L$ are established that can be accepted by
 a number of automata classes (deterministic one-counter machines that
 are 3-reversal-bounded, deterministic 2-counter machines that are 
 1-reversal-bounded, nondeterministic one-counter one-reversal-bounded machines), but taking any of the set of infixes, suffixes, or outfixes of $L$ produces languages that cannot be accepted by deterministic machines with
 an unrestricted pushdown and a fixed number of reversal-bounded counters.
 Hence, these deletion operations can create some very complex languages.
 It has been previously shown in \cite{EIMDeletion2015} though, 
 that the set of all
 infixes or suffixes of all deterministic one-counter one-reversal-bounded
 languages only produce deterministic reversal-bounded multicounter languages.
 Finally, the notion of $\varrho$-bounded-dense languages is defined and
 examined.

\section{Definitions}

In this section, some preliminary definitions are provided.

The set of non-negative integers is represented by $\mathbb{N}_0$. For $c \in \mathbb{N}_0$,
let $\pi(c)$ be $0$ if $c=0$, and $1$ otherwise.

We use standard notations for formal languages,
referring the reader to \cite{HU}.
The empty word is denoted by $\lambda$.
We use $\Sigma$ and $\Gamma$ to represent finite alphabets,
with $\Sigma^*$ as the set of all words over $\Sigma$
and $\Sigma^+ = \Sigma^* - \{\lambda\}$. 
For a word $w \in \Sigma^*$, if $w = a_1 \cdots a_n$
where $a_i \in \Sigma$, $1\leq i \leq n$, the length of $w$ is denoted by $|w|=n$,
and the reversal of $w$ is denoted by $w^R = a_n \cdots a_1$. Given a language $L\subseteq \Sigma^*$,
the complement of $L$ over $\Sigma^*$, $\Sigma^* - L$ is denoted by $\overline{L}$.

The definitions of deterministic and nondeterministic finite
automata, deterministic and nondeterministic pushdown
automata, deterministic Turing Machines, and instantaneous descriptions will be used from \cite{HU}.


Notation for variations of word operations
which we will use throughout the paper are presented next. 

\begin{defin}
\label{def:opGeneralize}
For a language $L \subseteq \Sigma^*$, 
the prefix, suffix, infix, and outfix operations, respectively, are defined
as follows:
\begin{center}
$\begin{array}{ll}
\ \pref(L) = \set{w \sst wx \in L, x \in \Sigma^* },	&\ \suff(L) = \set{w \sst xw \in L, x \in \Sigma^* },\\
\ \infx(L) = \set{w \sst xwy \in L, x,y \in \Sigma^* },	&\ \outf(L) = \set{xy \sst xwy \in L, w \in \Sigma^* }.
\\
\end{array}$
\end{center}
\end{defin}


Different types of density are now given. 
\begin{defin}
Let $\Sigma$ be an alphabet, and
$\varrho$ an operation from $\Sigma^*$ to $\Sigma^*$. Then 
$L \subseteq \Sigma^*$ is $\varrho$-dense if $\varrho(L) = \Sigma^*$.
\end{defin}


The reader is referred to \cite{Ibarra1978} and \cite{Baker1974}
for a comprehensive introduction to counter machines.
A {\em nondeterministic multicounter machine} is an automaton
which, in addition to having a finite set of states,
has a fixed number of counters.
At any point, the counters may be incremented, decremented,
or queried for equality to zero.
For our purposes, it will accept a word by final state.

Formally, a {\em one-way $k$-counter machine} is a tuple 
$M = (k,Q, \Sigma, \lhd, \delta,q_0, F)$,
where $Q, \Sigma, \lhd, q_0, F$ are respectively the set of states, input alphabet, right input end-marker, initial state (in $Q$) and  accepting states (a subset of $Q$).
The transition function $\delta$ (defined as in \cite{Chiniforooshan2012}) is a relation from 
$Q \times (\Sigma\cup \{\lhd\}) \times \{0,1\}^k$ to $Q \times \{{\rm S},{\rm R}\} \times  \{-1, 0, +1\}^k$, such
that if $\delta(q,a,c_1, \ldots, c_k)$ contains $(p,d,d_1, \ldots, d_k)$ and $c_i =0$ for some $i$, then $d_i \geq 0$ (to prevent negative values in any counter).
The symbols ${\rm S}$ and ${\rm R}$ indicate the direction that the input tape head moves, either {\em stay} or {\em right}. Further, $M$
is {\em deterministic} if $\delta$ is a partial function.
A configuration of $M$ is a $k+2$-tuple $(q, w\lhd , c_1, \ldots, c_k)$ representing
that $M$ is in state $q$, with $w\in \Sigma^*$ still to read as input, and $c_1, \ldots, c_k\in \mathbb{N}_0$ are the
contents of the $k$ counters. The derivation relation $\vdash_M$ is defined between configurations, where 
$(q, aw, c_1, \ldots , c_k) \vdash_M (p, w', c_1 + d_1, \ldots, c_k+d_k)$, if
$(p, d, d_1, \ldots, d_k) \in \delta(q, a, \pi(c_1), \ldots, \pi(c_k))$ where $d \in \{{\rm S}, {\rm R}\}$ and 
$w' =aw$ if $d={\rm S}$, and $w' = w$ if $d={\rm R}$. Let $\vdash^*_M$ be the reflexive, transitive closure of $\vdash_M$.
A word $w\in \Sigma^*$ is accepted by $M$ if 
$(q_0, w\lhd, 0, \ldots, 0) \vdash_M^* (q, \lhd, c_1, \ldots, c_k)$, for some $q \in F$, and 
$c_1, \ldots, c_k \in \mathbb{N}_0$. The language accepted by $M$, denoted by $L(M)$, 
is the set of all words accepted by $M$.
Furthermore, $M$ is $l$-reversal-bounded if it operates in such a way that in every accepting computation, the count on each counter alternates between increasing and decreasing at most $l$ times.
%

We will use the following notations for families of languages
(and classes of one-way machines): 
\begin{enumerate}
\item
$\NCM(k,l)$ for  
nondeterministic $l$-reversal-bounded $k$-counter languages,
\item
$\NCM = \bigcup_{k,l \geq 0} \NCM(k,l)$,
\item
$\NCA$ for nondeterministic 1-counter languages (no reversal bound),
\item
$\NPDA$ for
nondeterministic pushdown languages,
\item
$\NPDA(l)$ for nondeterministic $l$-reversal-bounded pushdown languages,
\item
$\NPCM$ for languages accepted by
nondeterministic machines with one unrestricted pushdown
and a fixed number of 
reversal-bounded counters.
\end{enumerate}
\noindent
For each of the above, replacing
N with D gives the deterministic variant.

It is easy to show that a counter that makes $l \ge 1$
reversals can be simulated by $\ceil{\frac{l+1}{2}}$
1-reversal-bounded counters \cite{Ibarra1978}.
So, e.g., for each $l \geq 1$,
$\DCM(1,l) \subseteq \DCM(\ceil{\frac{l+1}{2}}, 1)$
and thus $\DCM(1,3) \subseteq \DCM(2,1)$.
Thus the undecidability results for machines
with $k$ $l$-reversal-bounded counters
also carry over to machines with $k\ceil{\frac{l+1}{2}}$
1-reversal-bounded counters
(e.g., $ \DCM(1,3)$ to $\DCM(2,1)$).
%
%

We give some examples below to illustrate the workings of the
reversal-bounded counter machines.

\begin{exam}
Let $L = \{a^ib^ja^ib^j \mid i,j \ge 1\}$.  This language
(which is not context-free) can be
accepted by a $\DCM(2,1)$ which, when given
$a^ib^ja^kb^l$, reads the first segment
$a^ib^j$ and stores $i$ and $j$ in counters
$C_1$ and $C_2$, respectively. Then, it
reads the next segment $a^kb^l$ and verifies
that $i = k$ and $j = l$, by decrementing 
$C_1$ (resp., $C_2$) when reading $a^k$ (resp., $b^l)$.
\end{exam}

\begin{exam}
\label{exam2}
The Post Correspondence Problem (PCP) \cite{Post} is the problem of
deciding, given two $n$-tuples of strings $(X,Y)$, where $X= (x_1, \ldots , x_n)$
and $Y = (y_1,  \ldots, y_n)$ with each $x_i, y_i  \in \Sigma^+, 1 \leq i \leq n$, whether it
has a solution, i.e., whether there exists $i_1, \ldots, i_k, k \ge 1, 1 \leq i_l \leq n,$ for $1 \leq l \leq k$ such that 
$x_{i_1}  \cdots x_{i_k} = y_{i_1} \cdots y_{i_k}$.  
It is known that PCP is  undecidable for $|\Sigma| \ge 2$ \cite{Post}.

Now consider the following variation of PCP, called Permuted PCP:
Given $X$ and $Y$ as above, does
there exist $k \ge 1, I = (i_1, \ldots, i_k)$, 
and $J = (j_1, \ldots, j_k)$, $ 1 \leq i_l \leq n, 1 \leq j_l \leq n$,
for $1 \leq l \leq k$, $I$ is a permutation
of $J$, such that $x_{i_1} \cdots x_{i_k} = y_{j_1} \cdots y_{j_k}$.  It was 
shown in \cite{IbarraKim} that Permuted PCP is decidable using a 
restricted model of a multihead pushdown automaton whose
emptiness problem is decidable. Below, we use the technique
in \cite{IbarraKim} to show this result using $\NCM$s. 

Given $(X,Y)$, let $L(X,Y)  = \{ w \mid $ for some $k \ge 1,
w = x_{i_1} \cdots x_{i_k} = y_{j_1} \cdots y_{j_k}, (i_1,  \ldots, i_k)$
is a permutation of $(j_1, \ldots, j_k) \}$.   We can construct an
$\NCM$ $M$ to accept $L(X,Y)$.  $M$ has the tuples
$X = (x_1, \ldots, x_n)$ and $Y = (y_1, \ldots, y_n)$ in its
finite-state control and has $2n$ 1-reversal-bounded
counters $C_1, \ldots, C_n, D_1, \ldots, D_n$.  $M$
operates as follows, given input $w$.  It reads
the input $w$ and, in parallel, nondeterministically
guesses two decompositions of $w$:
$w = x_{i_1}\cdots x_{i_r}$  and $w = y_{j_1} \cdots y_{j_s}$
while incrementing counter $C_p$ every time it guesses
and verifies that $x_{i_t} = x_p$  ($t = 1, \ldots, r$) and incrementing
$D_p$ every time it guesses and verifies that $y_{j_t} = y_p$
($t = 1, \ldots, s$).  When $M$ reaches the end of the input,
it decrements the counters and accepts if and
only if  $C_i = D_i$ for $1 \le i \le n$.   Since the
emptiness problem for $\NCM$ is decidable \cite{Ibarra1978},
it follows that Permuted PCP is decidable.
\end{exam}

\begin{exam}
Let $(X,Y)$ be as above, and let
$L'(X,Y)  = \{ w ~|~ $ for some $k \ge 1,
w = x_{i_1} \cdots x_{i_k}, w^R = y_{j_1} \cdots y_{j_k}, (i_1,  \ldots, i_k)$
is a permutation of $(j_1, \ldots, j_k) \}$.
We can construct an $\NPCM$ $M'$ with $2n$ 1-reversal-bounded counters
and whose stack makes only one reversal to accept $L'(X,Y)$.
The idea is for $M'$ to guess the decomposition of $w$,
$w = x_{i_1} \cdots x_{i_r}$ while storing the number
of $x_p$'s there are in string $w$ in counter $C_p$ 
and copying $w$ in the stack. When $M$ reaches
the end of the input, it pops the stack and 
guesses the decmposition $w^R$,
$w^R = y_{j_1} \cdots y_{j_s}$ and storing the
number of $y_p$'s there are in string $w^R$ 
in counter $D_p$.  Finally, $M'$ checks that
$C_i = D_i$ for $1 \le i \le n$.
\end{exam} 

\begin{exam}
Let $M$ be the $\NCM$ accepting the
language $L(X,Y)$ in Example \ref{exam2}.
We can attach a pushdown stack to $M$
and obtain a $\NPCM$ $M'$. Clearly, such a
machine can accept a language $L(X,Y) \cap L$,
where $L$ is a context-free language.
Since the emptiness problem for $\NPCM$ is decidable \cite{Ibarra1978},
it follows that it is decidable, given $(X,Y)$
and an $\NPDA$ $M''$, whether $L(X,Y) \cap L(M'') = \emptyset$.
\end{exam}

\section{Deciding Types of Density}

In addition to examining
decidability of $\varrho$-density, a variant is defined 
called $\varrho$-marked-density that differs from $\varrho$-density
only by an end-marker.
\begin{defin}
Let $\Sigma$ be an alphabet, $\$ \notin \Sigma$, $L \subseteq (\Sigma \cup \{\$\})^*$, and $\varrho$ be an
operation from $(\Sigma \cup \{\$\})^*$ to itself. Then $L$ is $\varrho$-marked-dense if $\$\Sigma^* \$ \subseteq \varrho(L)$.
\end{defin}
It is only the marker $\$$ that differs from the usual
$\varrho$-dense (i.e., $\Sigma^* \subseteq \varrho(L)$ if and only
if $\Sigma^* = \varrho(L)$ for $L \subseteq \Sigma^*$). Yet we will
see differences, as there are cases when the marked version is
undecidable when the unmarked version is decidable.

First, deciding if languages are prefix-dense will be examined. It was recently shown in
\cite{EIMDeletion2015} that $\DCM$ languages are closed under prefix.
The following is a result in that paper.
\begin{prop} \label{cor1}
For $L \in \DCM$, $\pref(L) \in \DCM$.
\end{prop}
A main result in that paper was in fact far more general, showing that
$\DCM$ is closed (with an effective construction) under right quotient with $\NPCM$ languages.
Combining this with the known decidability of the 
inclusion problem for $\DCM$ \cite{Ibarra1978}, the following 
two corollaries are obtained, by testing if $\Sigma^* \subseteq \pref(L)$:
\begin{cor}
For $L_1, L_2 \in \DCM$,
it is decidable whether $\pref(L_1) \subseteq L_2$ and whether $L_1 \subseteq \pref(L_2)$.
\end{cor}

\begin{cor}
It is decidable whether a given $\DCM$ language is prefix-dense, and
prefix-marked-dense.
\end{cor}

This result is essentially the same for $\DPDA$, since it was shown that $\DPDA$ is closed under prefix \cite{GinsburgDPDAs}. Then,
for prefix-density, it suffices to determine if $\Sigma^*$ is equal to the prefix closure. And for prefix-marked-density, it suffices
to determine if $\$\Sigma^*\$$ is equal to the prefix closure intersected with the regular language $\$\Sigma^*\$$ ($\DPDA$ is closed
under intersection with regular languages \cite{GinsburgDPDAs}). And in both cases, the equality problem is decidable for $\DPDA$ \cite{Senizergues}.
\begin{prop}
It is decidable whether a given $\DPDA$ language is prefix-dense, and
prefix-marked-dense.
\end{prop}

Next, it is shown in \cite{EIMDeletion2015} that the set of suffixes and infixes of a $\DCM(1,1)$ language is always in $\DCM$ (by sometimes
increasing the number of counters). 
From this, the following is obtained:
\begin{prop}
For $L_1, L_2 \in \DCM(1,1)$,
it is decidable whether $\inf(L_1) \subseteq L_2$ and whether $L_1 \subseteq \inf(L_2)$. It is also decidable whether $\suff(L_1) \subseteq L_2$ and whether $L_1 \subseteq \suff(L_2)$.
\end{prop}

\begin{cor}
It is decidable whether a $\DCM(1,1)$ language is infix-dense, suffix-dense, infix-marked-dense and suffix-marked-dense.
\end{cor}
This result will be improved shortly
using a more general machine class for infix-density, but not
for suffix-density, suffix-marked-density, or infix-marked-density.

Most undecidability proofs in this section use the halting
problem for Turing machines.
Let $U \subseteq\{a\}^*$ be a unary recursively enumerable language that
is not recursive, i.e., not decidable (such a $U$ exists \cite{Minsky}),
and let $Z$ be a deterministic Turing machine accepting $U$.
Assume that $Z$ accepts if and only if $Z$ halts.

Let $Q$ and $\Gamma$ be the state set and worktape alphabet of $Z$,
and $q_0 \in Q$ be the initial state of $Z$. 
Note that $a$ is in $\Gamma$. Let $\Sigma = Q \cup \Gamma \cup \{\#\}$.
Assume without loss of generality that if $Z$ halts, it does so in a 
unique final state $q_f \neq q_0$, and a  unique configuration, and that the initial state $q_0$ is never re-entered
after the initial configuration, and that the length of every halting computation is even.

The halting computation of $Z$
on the input $a^d$ (if it accepts) can be represented by the string 
$x_d = ID_1\#ID_2^R\# 
\cdots \#ID_{k-1}\#ID_k^R$ for some $k \ge 2$,
where $ID_1 = q_0a^d$ and $ID_k$ are the initial and
unique halting configurations of $Z$, and $(ID_1, ID_2,\cdots,~ ID_k)$
is a valid sequence of instantaneous descriptions (IDs, defined in \cite{HU}) of $Z$ on input $a^d$, i.e., 
configuration $ID_{i+1}$ is a valid successor of $ID_i$, and $k$ is even.

Let $d \ge 0$.
Let $T$ be all strings
$w$ of the form $ID_1 \# ID_2^R \#  \cdots \# ID_{k-1} \# ID_k^R$,
where $k \ge 2$, $ID_1 = q_0a^d$, and $ID_k$ is the halting 
configuration of $Z$, and $ID_i$ is any ID
of the Turing machine, $1 < i <k$.
Then $T$ is a regular language, and thus a DFA $M_T$ can be built accepting $T$,
and also one can be built accepting $\overline{T}$.
Let $L_{na}$ be all strings
$w\in T$ of the form $ID_1 \# ID_2^R \#  \cdots \# ID_{k-1}\# ID_k^R$,
where there is an $i$ such that 
$ID_{i+1}$ is not a valid successor of $ID_i$.
Indeed, if $ID_{i+1}$ is not a valid successor of $ID_i$, then this is
detectable by scanning the state of $ID_i$, the letter after the
state (symbol under the read/write head), and from these, the transition
of $Z$ applied to get the valid successor of $ID_i$ can be calculated, as with whether the $ID$ representing the valid successor to $ID_i$ should be shorter or longer by one symbol.
Then, there is some position $j$ of $ID_i$ such that
examining positions $j-2,j-1, j, j+1, j+2$ of
$ID_i$ and $ID_{i+1}$,
and the state of $ID_i$ and $ID_{i+1}$ is enough to imply that $ID_{i+1}$ is
not a valid successor.
Hence, let
$L_{na}(p)$ be the set of words 
$w\in T$ of the form $w=ID_1 \# ID_2^R \#  \cdots \# ID_{k-1} \# ID_k^R$,
where the $p$th character of $w$ is within the string $ID_i$
for some $i$ at position $j$ of $ID_i$ and
examining characters $j-2,j-1,j, j+1,j+2$ of $ID_i$ and $ID_{i+1}$ (if they exist), plus the states of both, and the letter after the state, 
implies that
$ID_{i+1}$ is not a valid successor of $ID_i$.
Thus, $\bigcup_{p\geq 1}L_{na}(p) = L_{na}$.

Let $L_d = L_{na} \cup \overline{T}$. Two lemmas are required for undecidability results.

\begin{lemma}
\label{lemmaNA}
$L_d = \Sigma^*$ if and only if $T\subseteq L_{na}$ if and only if $Z$ does not halt on $a^d$.
\end{lemma}
\begin{proof}
If $L_d=\Sigma^*$, then $T\subseteq L_{na}$, and if $T\subseteq L_{na}$ then
$T\cup \overline{T} = \Sigma^* \subseteq L_d$. Thus the first two are equivalent.

Assume $L_d = \Sigma^*$. Thus, every sequence of IDs in $T$ is in $L_{na}$, thus there is no sequence of IDs that halts on $a^d$.

Assume that $Z$ does not halt on $a^d$. Let $w \in \Sigma^*$. If $w \notin T$, then $w \in L_d$.
If $w \in T$, then $w$ does not represent an accepting computation, thus,
$w \in L_d$.
\end{proof}

Let $\%$ be a new symbol not in $\Sigma$, and let $\Sigma_{\%} = \Sigma \cup \{\%\}$.
\begin{lemma}
\label{lemmaDCM}
$\bigcup_{p \geq 1} \%^p L_{na}(p)$ and $\bigcup_{p \geq 1} \%^p \$ L_{na}(p) \$$ are both in $\DCM(1,3)$ and $\DCM(2,1)$.
Furthermore, $L_{na},\$L_{na}\$ \in \NCM(1,1)$.
\end{lemma}
\begin{proof}
We can construct a $\DCM(1,3)$ machine $M_{na}$ to accept the strings of 
$\bigcup_{p \geq 1} \%^p L_{na}(p)$ as follows:
when given $\%^p w $, it reads $\%^p$ and increments 
the counter by $p$.  It then decrements the counter and
verifies that when the counter becomes zero, the input head
is within some $ID_i$ (or $ID_i^R$ if $i$ is even).  
If $i$ is odd, $M_{na}$ then moves the input
head incrementing the counter until it reaches the $\#$
to the right of $ID_i$.
Let $j$ be the value of the counter.
$M_{na}$ then decrements the counter while
moving right on $ID_{i+1}^R$ and after reaching zero, verifying
that  $ID_{i+1}$
is not a valid successor of $ID_i$ (this is possible as $ID_{i+1}^R$ is
 reversed).  Similarly when $i$ is even.
In the same way, we can construct a $\DCM(1,3)$ machine to
accept $\bigcup_{p \geq 1} \%^p \$ L_{na}(p) \$$.
Both languages are in $\DCM(2,1)$ as $\DCM(1,3) \subseteq \DCM(2,1)$.

For $L_{na}$ (and $\$L_{na}\$$), it is possible to nondeterministically
guess the position $p$, and then when within $ID_i$, verify using
the counter once that $ID_{i+1}$ is not a valid successor to $ID_i$. 
 \end{proof}
This is similar to the technique from \cite{Baker1974} to show undecidability of universality for $\NCM(1,1)$.

Most of the undecidability results in this section build off of the above two lemmas, the input $a^d$, the languages $T, L_{na}$, etc.

\begin{prop}
\label{prop:suffSubset}
Let $\Sigma$ be an alphabet.
\begin{enumerate}
\item
It is undecidable to determine, given $L \in \NCM(1,1)$, whether
$L$ is $\varrho$-marked-dense, for $\varrho \in  \{\suff, \inf, \pref \}$.
\item
It is undecidable to determine, given $L \in \DCM(1,3)$, whether
$L$ is $\varrho$-marked-dense, for $\varrho \in \{\suff, \inf\} $.
\item
It is undecidable to determine, given $L \in \DCM(2,1)$, whether
$L$ is $\varrho$-marked-dense, for $\varrho \in \{\suff, \inf \}$.
\end{enumerate}
\end{prop}
\begin{proof}
For part 1, we can accept $L' = \$ L_{na}\$ \cup \$\overline{T} \$ 
\subseteq (\Sigma \cup \{\$\})^*$ in $\NCM(1,1)$ since
$\$\overline{T}\$$ is a regular language (the complement is over $\Sigma^*$).

Then
$\$\Sigma^* \$ \subseteq \inf(L')$ (resp., $\$\Sigma^*\$ \subseteq \suff(L')$,
$\$\Sigma \$ \subseteq \pref(L')$) if and only if $\$\Sigma^* \$ = L'$
if and only if $L_d = \Sigma^*$, which we
already know is true
if and only if $Z$ does not halt on $a^d$ by Lemma \ref{lemmaNA}, which is undecidable.

For parts 2 and 3, we instead use
$L' =  \bigcup_{p \geq 1} \%^p \$ L_{na}(p) \$ \cup \$\overline{ T }\$$,
the complement $\overline{T}$ is over $\Sigma_{\%}^* = (\Sigma \cup \{\%\})^*$ here, so it will also contain
any word with $\%$ in it to allow for marked-density to be with
$L' \subseteq (\Sigma_{\%} \cup \{\$\})^*$ where the goal is to decide
whether $\$\Sigma_{\%}^*\$ \subseteq \inf(L')$. Then $L'$
is in $\DCM(1,3) \cap \DCM(2,1)$ by Lemma \ref{lemmaDCM} and since $\DCM(k,l)$ is closed under union with regular languages, for every $k,l$
\cite{Ibarra1978}. And
$\$\Sigma_{\%}^*\$ \subseteq \inf(L')$ if and only if 
$\$\Sigma^* \$ \subseteq \infx(L')$ (since $\overline{T}$ contains all 
words with at least one $\%$) if and only if  $L_d = \Sigma^*$. The proof in
the case of the suffix operation is similar.
 \end{proof}

The proof for the outfix operation is similar.
\begin{prop} \label{prop2}
It is undecidable, given $L \in \NCM(1,1)$, whether
$L$ is $\outf$-marked-dense. Similarly with $L\in \DCM(2,1)$,
and $L \in \DCM(1,3)$.
\end{prop}
\begin{proof}
For $L \in \NCM(1,1)$, we modify the language $L'$ in the proof
of Part 1 of Proposition \ref{prop:suffSubset}. 
So $L' = \% \$ L_{na}\$ \cup \% \$\overline{T}\$$ ($\overline{T}$ over
$\Sigma_{\%}^*$).
For the
other classes, $L'$ in the proofs of parts 2, 3
also work for $\outf$.
 \end{proof} 


It follows from Propositions \ref{prop:suffSubset} and \ref{prop2} that
$\DPDA(3)$ has an undecidable
$\varrho$-marked-density problem for suffix, infix, and outfix.
The following shows that they are also undecidable
for $\DPDA(1)$.
\begin{prop}
For $\varrho \in \{\suff, \inf,\outf\}$, it is undecidable given $L \in \DPDA(1)$, whether $L$ is $\varrho$-marked-dense.
\end{prop}
\begin{proof} The problem of whether the
intersection of two $\DPDA(1)$ languages is empty
is undecidable \cite{Baker1974}.
Let $L_1, L_2 \in \DPDA(1)$. Then 
$L_1 \cap L_2 = \emptyset$ if and only if
$\overline{L_1 \cap L_2} = \Sigma^*$ if and only if
$\overline{L_1} \cup \overline{L_2} = \Sigma^*$ if and only if
$\$\Sigma^* \$ \subseteq \$ \overline{L_1}\$ \cup \$\overline{L_2}\$$.

Let $L' = \% \$ \overline{L_1} \$ \cup \$ \overline{L_2} \$ \cup \$ \Sigma_{\%}^* \% \Sigma_{\%}^* \$$ (here, the complements are over $\Sigma^*$).
Note that $L' \subseteq (\Sigma_{\%} \cup \{\$\})^*
= (\Sigma \cup \{\%,\$\})^*$.
$L'$ is in $\DPDA(1)$ since $\DPDA(1)$ is closed under complement, the union of the first two sets is a $\DPDA(1)$ language (if \% is the first letter then simulate the first set, otherwise simulate the second), and the third one is regular and $\DPDA(1)$ is closed under union with regular sets.

Then
$\$\Sigma_{\%}^* \$ \subseteq \inf(L')$ if and only if
$\$\Sigma_{\%}^* \$ \subseteq \$\overline{L_1}\$ \cup \$\overline{L_2}\$ \cup \$\Sigma_{\%}^* \% \Sigma_{\%}^*\$$ if and only if
$\$\Sigma^*\$ \subseteq \$ \overline{L_1} \$ \cup \$ \overline{L_2} \$$, which we know is undecidable. The proof is identical for suffix,
as with outfix after preceding each word in $L'$ by an additional $\%$.
 \end{proof}

Next, $\varrho$-density instead of $\varrho$-marked-density will be considered; specifically, the question of whether it is decidable
to determine if a language $L$
is $\varrho$-dense ($\varrho(L) = \Sigma^*$) for various operations and languages. For suffix-density,
undecidability occurs for the same families as
for marked-suffix-density. The proofs will again build on the Turing Machine
$Z$, input $a^d$, and languages $L_{na},T$, etc.

\begin{prop} \label{prop4}
Let $L \in \DCM(1,3)$. It is undecidable to determine if $L$ is suffix-dense.
Similarly for $L \in \DCM(2,1)$ and $L \in \NCM(1,1)$.
\label{suffixdense}
\end{prop}
\begin{proof}
Let $L_1' = \{ \%^p ux \mid u \in \Sigma \Sigma_{\%}^*, p = |u| + p', x \in L_{na}(p')\}$,
$L_2' = \overline{\Sigma_{\%}^* T}$ (over $\Sigma_{\%}^*$),  
and $L' = L_1' \cup L_2'$.
Then $L_1' \in \DCM(1,3)$ as one can build $M_1' \in \DCM(1,3)$ by
adding $p$ to the counter until hitting a letter that is not $\%$. 
Then as $M_1'$ reads the remaining input in $\Sigma \Sigma_{\%}^*$, for every character read, it decreases the counter, and each time $M_1'$ hits state $q_0$ (which could be the beginning of a word in $T$), it runs
$M_T$ (the DFA accepting $T$) in parallel to check if the suffix starting at this position is in $T$. If it hits $q_0$ more than once, it can stop previous simulations of $M_T$ and start a new simulation.
However, it is only required that a suffix of the input is in $T$. If the counter empties while $M_T$ is running in parallel, then let
$ux$ be the input, where $u$ is the input before
reaching $q_0$ in the current run of $M_T$, and $x$ be the input from
$q_0$ to the end. Then $M_1'$ tries to verify that $x \in L_{na}(p')$,
where $p = |u| + p'$. When the counter reaches
$0$, $M_1'$ has subtracted $1$ from the counter the length of $u$ plus
$p - |u| = p'$ times. Thus, $M_1'$ can continue the
simulation of $M_{na}$ from Lemma \ref{lemmaDCM} from when
the counter reaches $0$, thereby verifying that
$x \in L_{na}(p')$ (and $x \in T$). Then $L' \in \DCM(1,3)$
as each $\DCM(k,l)$ is closed under union with regular languages
\cite{Ibarra1978}.
Then also $L'$ must be in $\DCM(2,1)$.

It will be shown that $\suff(L') = \Sigma_{\%}^*$ if and only if $T \subseteq L_{na}$, which
is enough by Lemma \ref{lemmaNA}.

``$\Leftarrow$''
Assume $T\subseteq L_{na}$.
Let $w \in \Sigma_{\%}^*$. 

Assume that there exists a (potentially not proper) suffix of $w$ in $T$. Then $w = ux, x\in T, u \in \Sigma_{\%}^*$. 
Then $x\in L_{na}$, by assumption. 
Then there exists $p$ such that $\%^p x \in L_{na}(p), x \in T$ and so 
$\%^{p'} a u  x \in L_1', au \in \Sigma \Sigma_{\%}^*$, where $p'= p+|au|$. Thus $ux=w \in \suff(L_1')$. 

Assume that there does not exist a suffix of $w$ in $T$. Then $w \in L_2'$, and $w \in \suff(L')$.

``$\Rightarrow$''
Assume $\suff(L') = \Sigma_{\%}^*$. 
Let $w \in T$. Then $w \in \suff(L')$. 
Then there exists $\%^p u w \in L_1'$. This implies there exists $p'$ such that $w \in L_{na}(p') \subseteq L_{na}$.

The case for $\NCM(1,1)$ is similar except using 
$L_1' = \{ux \mid x \in L_{na}, u \in \Sigma^* \}$ 
and $L_2' = \overline{\Sigma^* T}$,
and $L' = L_1' \cup L_2' \subseteq \Sigma^*$, as $u$ can be nondeterministically
guessed without using the counter.
 \end{proof}

\begin{cor} \label{cor5}
For $L \in \NCM(1,1)$, the question of whether $L$ is prefix-dense is undecidable.
\end{cor}
\begin{proof}
It is known that $\NCM(k,l)$ is closed under reversal for each $k,l$. 
Also, $\pref(L^R) = \Sigma^*$ if and only if $\suff(L) = \Sigma^*$.
 \end{proof}

We are able to extend the undecidability results to infix-density, but only by using one unrestricted counter and with nondeterminism.

\begin{prop}
Let $L\in \NCA$. The question of whether $L$ is infix-dense is undecidable.
\end{prop}
\begin{proof}
Let $L' = \overline{(\Sigma^* T \Sigma^*)} (L_{na} \overline{(\Sigma^* T \Sigma^*)})^* \subseteq \Sigma^*$. 
It is clear that $L' \in \NCA$.
We will show that $T \subseteq L_{na}$ if and only if $\inf(L') = \Sigma^*$.

``$\Rightarrow$''
Assume $T \subseteq L_{na}$. Let $ w\in \Sigma^*$. If $w \in \overline{(\Sigma^* T \Sigma^*)}$, then $w \in L' \subseteq \inf(L')$.
Assume $w \notin \overline{(\Sigma^* T \Sigma^*)}$. Then 
$w \in \Sigma^* T \Sigma^*$. Then $w = u_0 v_1 u_1 \cdots u_{n-1} v_n u_n$, where $n \geq 1, v_1, \ldots, v_n \in T$, and $u_0, \ldots, u_n \notin \Sigma^* T \Sigma^*$, and so
$u_0, \ldots, u_n \in \overline{(\Sigma^* T \Sigma^*)}$. Also, 
$T \subseteq L_{na}$, and therefore $v_1, \ldots, v_n \in L_{na}$ and
$w \in L' \subseteq \inf(L')$.

``$\Leftarrow$''
Assume $\inf(L') = \Sigma^*$. Let $w\in T$. Then $w \in \inf(L')$. Since $w \in \inf(L') \cap T$, then $x = uwv \in L'$. Then
$x = u_0 v_1 u_1 \cdots u_{n-1} v_n u_n$, where $n \geq 1, v_1, \ldots, v_n \in L_{na}$, and $u_0, \ldots, u_n \in \overline{\Sigma^* T \Sigma^*}$.
If $w$ is an infix of $u_i$, for some $i$, then $u_i \in \Sigma^*T \Sigma^*$, a contradiction. If $w$ overlaps with $v_i$ for some $i$, then it must be exactly one $v_i$ by the structure of $T$ (initial and final states are only used once at beginning and end of words in $T$). Then $w\in L_{na}$.
 \end{proof}

The same undecidability is obtained with determinism, but an unrestricted
pushdown automaton is used.
\begin{prop}
Let $L\in \DPDA$. The question of whether $L$ is infix-dense is undecidable.
\end{prop}
\begin{proof}
Let $\Sigma_1 = \Sigma \cup \{\%, e, \cent\}$. Let
$$L' = \{ \begin{array}[t]{l} r_m r_{m-1} \cdots r_1 \cent u_0 y_1 u_1 \cdots y_m u_m \mid m \geq 0,
u_i \in \overline{\Sigma_1^* T \Sigma_1^*}, 0 \leq i \leq m,\\
								\hspace{.5cm} y_j \in T, r_j = \%^{p_j} e^{q_j},
								 q_j = |u_{j-1}|, y_j \in L_{na}(p_j) 
								 \mbox{~for~} 1 \leq j \leq m \}.\end{array} 
	$$								
(In the above set, the complementation is over $\Sigma_1^*$.)
First, $L'$ can be accepted by a $\DPDA$ as follows: create $M'$ 
that reads $r_m \cdots r_1$ and pushes each
symbol onto the pushdown, which is now (with bottom of pushdown
marker $Z_0$)
$$Z_0 \%^{p_m} e^{q_m} \cdots \%^{p_1} e^{q_1}.$$
Then for each $\%^{p_j} e^{q_j}$ on the pushdown from $1$ to $m$, 
$M'$ reads one symbol at a time from the input while popping one $e$, while in parallel verifying $u_{j-1} \in \overline{\Sigma_1^* T \Sigma_1^*}$. Then $M'$
verifies that $y_j \in L_{na}(p_j)$ as in Lemma \ref{lemmaDCM} (by 
popping $\%^{p_j}$ one symbol at a time until zero and then 
pushing on the pushdown simulating the counter). Finally $M'$ verifies 
$u_m \in \overline{\Sigma_1^* T \Sigma_1^*}$.

We claim that $\inf(L') = \Sigma_1^*$ if and only if $T \subseteq L_{na}$.

Assume $T \subseteq L_{na}$. Let $w \in \Sigma_1^*$. We will show $w \in \infx(L')$. Let $w = u_0 y_1 u_1 \cdots u_{m-1} y_m u_m$, where
$m \geq 0, y_1, \ldots, y_m \in T, u_0, \ldots, u_m \in \overline{\Sigma_1^* T \Sigma_1^*}$. Then for each $y_j, 1 \leq j \leq m, y_j \in L_{na}(p_j)$, for some $p_j$, and thus there exists 
$q_j$ such that
$q_j = |u_{j-1}|$.
Thus, $\%^{p_m} e^{q_m} \cdots \%^{p_1} e^{q_1} \cent w \in L',$ and $w \in \infx(L')$.

Assume $\infx(L') = \Sigma_1^*$. Let $w\in T$. Then there must exist $x,y$ such that $z = xwy \in L'$. 
Then $z = u_0 y_1 u_1 \cdots u_{m-1} y_m u_m$, where $y_1, \ldots, y_m \in T, u_0, \ldots, u_m \in \overline{\Sigma_1^* T \Sigma_1^*}$.
Necessarily, one of $y_1, \ldots, y_m$, $y_i$ say, must be $w$. This implies $w = y_i \in L_{na}(p_i)$, for some $p_i$. Hence, $w \in L_{na}$.
 \end{proof}

In contrast to the undecidability of marked-infix-density and suffix-density
for $\DCM(1,3)$ and $\NCM(1,1)$, for infix-density on reversal-bounded nondeterministic pushdown automata, it
is decidable. The main tool of the proof is the
known fact that the language of all words over the pushdown alphabet that can appear on the pushdown in an accepting computation is a regular language \cite{CFHandbook}.
\begin{prop} \label{prop7}
It is decidable, given $L$ accepted by a one-way reversal-bounded $\NPDA$,
whether $L$ is infix-dense.
\end{prop}
\begin{proof}
Let $M = (Q,\Sigma, \Gamma, \delta,q_0, Z_0,F)$ be a pushdown automaton that 
accepts by final state and never pops $Z_0$. The pushdown is
said to be empty if $Z_0$ is at the top of the pushdown. Also, assume
$M$ makes at most $l$ switches between increasing and decreasing the size of the pushdown. Assume without loss of generality that $Q$ is partitioned into sets
$Q = \bigcup_{0 \leq i \leq l} Q_i \cup \bigcup_{0 \leq i \leq l} \bar{Q_i}$, where
$Q_i$ consists of all states defined on or after the $i$th reversal and before the $i+1$st on a non-empty pushdown, and $\bar{Q_i}$ is the same on an empty pushdown. Also, 
assume without loss of generality that all transitions either
push one letter, keep the stack the same, or pop one letter, and there are no $\lambda$-transitions
that do not change the pushdown.

Let $q \in Q_i \cup \bar{Q_i}$ for some $i$. Let $L_q$ be the language
$$\{ w \mid (q_0, uwv, Z_0) \vdash_M^* (q,wv,\alpha) \vdash_M^* 
(q, v,\beta) \vdash_M^* (q_f, \lambda, \gamma), q_f \in F, \alpha,\beta, \gamma \in \Gamma^*\}.$$
It will be shown that $L^q$ is a regular language, for all $q\in Q$.
Let $h_{\Sigma}$ be the homomorphism from $(\Sigma\cup \Gamma)^*$ to $\Sigma$ that erases all letters of $\Gamma$ that fixes all letters of $\Sigma$, and let
$h_{\Gamma}$ be the homomorphism that erases all letters of $\Sigma$ and fixes
all letters of $\Gamma$.

Consider the languages $$Acc(q) = \{ \alpha \in \Gamma^* \mid (q_0, u, Z_0) \vdash_M^* (q, \lambda,\alpha), u \in \Sigma^*\}.$$ and
$$co\mbox{-}Acc(q) = \{\beta \in \Gamma^* \mid (q, v,\beta) \vdash_M^* (q_f, \lambda, \gamma), v \in \Sigma^*, q_f \in F\}.$$ It is shown in \cite{CFHandbook} that both of these languages are in fact regular languages. Moreover, the proofs contain effective constructions.

Then consider $L^q, q \in \bar{Q_i}$. If either $Acc(q)$ or $co\mbox{-}Acc(q)$ are empty, then so is $L^q$. If both are non-empty, then $L_q$
can be accepted by simulating $M$ with an NFA which can be done since all transitions are on $Z_0$.

Consider $L^q, q \in Q_i, i$ even. Thus, there are no decreasing transitions or transitions on an empty pushdown defined between states $q$ and $q$.
Create an interim NFA $M'$ accepting an interim language $L^q_1 \subseteq \Gamma^+ (\Sigma \cup \Gamma)^*$, where $M'$ does the following in parallel:
\begin{itemize}
\item Nondeterministically guesses a partition of the input into $\alpha y$,
where $\alpha \in \Gamma^+, y \in (\Sigma \cup \Gamma)^*$, verifies that 
$\alpha \in Acc(q)$, and $M'$ also remembers the last letter of $\Gamma$ in $\alpha$, and as it reads $y$, continues to remember the previous symbol from
$\Gamma$ encountered. Then it reads the remaining input
$y$ and starting in state $q$, simulates $M$ as follows: if the next letters are 
$a \in \Sigma$ followed by $d\in \Gamma$, then $M'$ simulates (just
by reading these letters $ad$ and switching states appropriately)
a transition that reads $a$ that is defined on the remembered pushdown letter
while pushing $d$ on the pushdown. 
Otherwise, if the next letter is $a\in \Sigma$ and a letter from $\Gamma$
does not follow, then $M'$ simulates a transition that reads $a$ and does not push on the remembered pushdown letter. If instead the next letter is $d\in \Gamma$,
then $M'$ simulates a pushing of $d$ on $\lambda$ input on the 
remembered pushdown letter.
At the end of the input, the simulated machine $M$
must be in state $q$. 
\item Also, $M'$ reads the input, and if 
$h_{\Gamma}(\alpha y)= \beta$ (this is $\alpha$ plus the word from $\Gamma$  shuffled into $y$), then $M'$ verifies that $\beta \in co\mbox{-}Acc(q)$.
\end{itemize}

\begin{claim}
$h_{\Sigma}(L^q_1) = L^q$.
\end{claim}
\begin{proof}
``$\subseteq$'' Let $s\in h_{\Sigma}(L_1^q)$. Thus, there exists
$t \in (\Sigma \cup \Gamma)^*$ such that $h_{\Sigma}(t) = s$ and
$t \in L_1^q$. Then $t = \alpha y, \alpha \in \Gamma^+, 
y \in (\Sigma \cup \Gamma)^*$, where $M'$
verifies $\alpha \in Acc(q)$.
Then $(q_0, u, Z_0) \vdash_M^* (q, \lambda, \alpha)$ for some $u \in \Sigma^*$.
Then, on each letter of $y$, $M'$ simulates $M$ on the last letter read from 
$\Gamma$, reading each letter from $\Sigma$ and pushing each letter from 
$\Gamma$ read starting and finishing in state $q$. Thus, 
$$(q, h_{\Sigma}(y), \alpha) \vdash_M^* (q, \lambda, \alpha h_{\Gamma}(y))
= (q, \lambda, h_{\Gamma}(\alpha y)).$$
Then since $M'$ verified $h_{\Gamma}(\alpha y) \in co\mbox{-}Acc(q)$,
this implies $(q, v, h_{\Gamma}(\alpha y)) \vdash_M^* (q_f, \lambda, \gamma),
q_f \in F, v \in \Sigma^*$. Hence $h_{\Sigma}(y) \in L^q$ and $s = h_{\Sigma}(y)$.

``$\supseteq$''
Let $s \in L^q$. Thus,
$$(q_0, usv, Z_0) \vdash_M^* (q, sv, \alpha) \vdash_M^* (q,v,\beta)
\vdash_M^* (q_f, \lambda, \gamma),$$
$q_f \in F$. Then $\beta = \alpha \mu$, for some $\mu\in \Gamma^*$.
Let
$$(p_0 = q, s_0 = sv, \alpha_0 = \alpha) \vdash_M (p_1, s_1, \alpha_1)
\vdash_M \cdots \vdash_M (p_n = q, s_n= v, \alpha_n = \beta),$$
be the derivation above between states $q$ and $q$ via transitions $t_1, \ldots, t_n$
respectively. Let $y$ be obtained by examining each $t_i$ in order, from $1$
to $n$, and concatenating $a \in \Sigma$ if $t_i$ consumes $a$ ($\lambda$
otherwise), $d\in \Gamma$ if $t_i$ pushes $d$ ($\lambda$ otherwise).
Then $\mu = h_{\Gamma}(y)$. We will show $\alpha y \in L(M')$. Indeed,
$\alpha \in Acc(q)$, and $M'$ simulates $M$ when reading $y$ ending in state 
$q$, and verifies that $\beta = \alpha \mu \in co\mbox{-}Acc(q)$.
Thus, $s \in h_{\Sigma}(L_1^q)$.
 \end{proof}

Similarly, if $L^q, q \in Q_i$, $i$ odd, then there are no pushing transitions or transitions on empty pushdown between states $q$ and $q$.
Then create an interim NFA $M'$ accepting an interim language $L^q_2 \subseteq (\Sigma \cup \Gamma)^* \Gamma^+$, where $M'$ does the following in parallel:
\begin{itemize}
\item Nondeterministically guesses a partition of the input into $y \alpha$, where $\alpha \in \Gamma^+$, and simulates $M$ from state $q$.
While reading a letter $a \in \Sigma$ followed by a letter from $\Gamma$, $M'$ simulates
a transition of $M$ that reads input letter $a \in \Sigma$ and pops the letter from 
$\Gamma$. Otherwise, if reading $a\in \Sigma$ not followed by a letter from
$\Gamma$, then $M'$ guesses the next letter $d\in \Gamma$ that will appear in
$y\alpha$, simulate a transition reading $a$ with $d$ on top of the pushdown
that does not change the pushdown, and eventually verify that $d$ is the
next letter of $\Gamma$. If reading $d\in \Gamma$ only, then $M'$ simulates
a $\lambda$-transition that pops $d\in \Gamma$. Then, when reaching the end of $y$, it verifies that the simulated machine is in state $q$, and verifies that the remaining input $\alpha \in co\mbox{-}Acc(q)^R$ 
(as $co\mbox{-}Acc(q)$ is a regular language, so is its reversal).
\item Also, $M'$ reads the input and if 
$h_{\Gamma}(y\alpha) = \beta$, then $M'$ verifies that $\beta \in Acc(q)^R$.
\end{itemize}

\begin{claim}
$h_{\Sigma}(L^q_2) = L^q$.
\end{claim}
\begin{proof}
``$\subseteq$''
Let $s \in h_{\Sigma}(L_2^q)$. Thus, there exists 
$t \in (\Sigma \cup \Gamma)^*$ such that $h_{\Sigma}(t) = s$ and
$t \in L_2^q$. Then $t = y \alpha, \alpha \in \Gamma^+, 
y \in (\Sigma \cup \Gamma)^*$. On each
letter of $y$, $M'$ simulates $M$ from state $q$ on a top-of-pushdown letter that is
nondeterministically guessed, then later verified when hitting the next 
letter of $\Gamma$ (the topmost symbol of the pushdown), on input 
letters from $\Sigma$ that are read, and reading letters from $\Gamma$
that are popped, ending in $q$.
Further, $M'$ verifies that $\alpha \in co\mbox{-}Acc(q)^R, 
h_{\Gamma}(y\alpha) = \beta \in Acc(q)^R$. Indeed,
letters are read in $M'$ in the same order that they are popped,
which is reversed from the languages $Acc(q)$ and $co\mbox{-}Acc(q)$,
which are the words that appear on the pushdowns from the bottom
towards the top.
Hence, there exists $u,v \in \Sigma^*$ such that
$$(q_0, usv, Z_0) \vdash_M^* (q, sv, \beta^R) \vdash_M^* (q,v,\alpha^R)
\vdash_M^* (q_f, \lambda, \gamma),$$ $q_f \in F$
since $\beta^R = h_{\Gamma}(\alpha^R y^R)$. Hence, $s \in L^q$.

``$\supseteq$''
Let $s \in L^q$. Thus,
$$(q_0, usv, Z_0) \vdash_M^* (q,sv,\beta^R) \vdash_M^* (q,v,\alpha^R)
\vdash_M^* (q_f, \lambda, \gamma),$$ $q_f \in F$. Therefore,
$\beta^R = \alpha^R \mu^R $ for some $\mu \in \Gamma^*$.
Let $$(p_0 = q, s_0 = sv, \alpha_0 = \beta^R) \vdash_M (p_1, s_1, \alpha_1)
\vdash_M \cdots \vdash_M (p_n = q, s_n= v, \alpha_n = \alpha^R)$$
be the derivation between the two configurations with $q$ above,
via transitions $t_1, \ldots, t_n$ respectively. Let $y$ be
obtained by examining each $t_i$ in order from $1$ to $n$ and 
concatenating $a \in \Sigma$ if $t_i$ consumes $a$ ($\lambda$ otherwise),
and $d\in \Gamma$ if $t_i$ pops $d$ ($\lambda$ otherwise). Then
$\mu^R = h_{\Gamma}(y)$. We will show $y \alpha^R \in L(M')$.
Indeed, $M'$ verifies $\beta = \mu \alpha \in Acc(q)^R$ and in parallel, $M'$ simulates $M$ from $q$, reading $y$ ending in state $q$ and 
verifies that $\alpha \in Acc(q)^R$. Further,
$s = h_{\Sigma}(y \alpha)$ and hence, $s \in h_{\Sigma}(L_2^q)$.

\
\end{proof}

Hence, $L^q$ is regular for all $q \in Q$ since regular languages are
closed under homomorphism.
Let $L' = \bigcup_{q \in Q} L^q$, which is also regular.

To conclude, it will be shown that $\inf(L) = \Sigma^*$ if and only if $\inf(L') = \Sigma^*$.

Assume that $\inf(L) = \Sigma^*$. Let $w \in \Sigma^*$. Then $w \in \inf(L)$. Consider $w' = w^{|Q|+1}$. Then $w' \in \inf(L)$, and by the pigeonhole
principal, an entire copy of $w$ has to be read between some state $q$ and itself. Then $w\in \inf(L^q)$.

The converse is trivial since $\inf(L') \subseteq \inf(L)$.

\prbox
\end{proof}

Next, we briefly examine the  reverse containments when testing if
$\Sigma^* \subseteq \varrho(L)$ and $\$ \Sigma^* \$ \subseteq \varrho(L)$ for
density and marked-density. Here, it is checked whether it is decidable to test
$\varrho(L) \subseteq R$ for regular languages $R$. In fact, we will
show a stronger result.

\begin{prop} \label{prop3}
It is decidable, given $L_1 \in \NPCM$
and $L_2 \in \DCM$,
whether $\varrho(L_1) \subseteq L_2$, where 
$\varrho \in \{\suff, \inf, \pref, \outf\}$.
\end{prop}
\begin{proof}
It is easy to show that if $L_1$ is accepted by an $\NPCM$,
then $\varrho(L_1)$ can be accepted by an
$\NPCM$.  We can also construct a $\DCM$ that
accepts $\overline{L_2}$ \cite{Ibarra1978}.  We can then
construct an $\NPCM$ accepting $\varrho(L_1)  \cap  \overline{L_2}$
as $\NPCM$ is closed under intersection with $\NCM$.
The decidability of whether $\varrho(L) \subseteq L_2$  is equivalent
to the question of whether the $\NPCM$ accepting
$\varrho(L)  \cap  \overline{L_2}$ is empty, which is decidable, since
the emptiness problem for $\NPCM$s is
decidable \cite{Ibarra1978}.  
 \end{proof}

The languages used in the proofs of Lemmas
\ref{lemmaNA} and \ref{lemmaDCM} are used next to show
that $\varrho(L)$ does not belong in the same 
family as $L$, in general.

\begin{prop}
There is a language $L \in \DCM(1,3)$ (resp., $\NCM(1,1)$,\\
$\DCM(2,1))$ such that $\varrho(L)$  is not in $\DPCM$,
where $\varrho \in \{\suff, \inf, \outf \}$.
\end{prop}
\begin{proof}
We first give a proof for $\DCM(1,3)$.
Consider $L' = \bigcup_{p \geq 1} \%^p \$ L_{na}(p)\$ \in \DCM(1,3)$ by Lemma \ref{lemmaDCM}.
For $\varrho \in \{\suff, \inf, \outf \}$, we claim that  $\varrho(L')$ cannot be accepted by any $\DPCM$. We know $\$L_{na}\$ \subseteq \varrho(L')$.  
For suppose
$\varrho(L')$ can be accepted by a $\DPCM$ $M_1$.
Then, since the family of languages accepted by $\DPCM$s is
closed under complementation \cite{IbarraYen}, 
we can construct a $\DPCM$ $M_2$ accepting $\overline{L(M_1)}$.
	Now using $M_2$, an algorithm can be constructed
	to determine whether $\$T\$\not\subseteq \$L_{na}\$$, which
	we know is a subset of $\varrho(L')$.
	\begin{enumerate}
           \item 	
	   Consider $T$, which can be accepted by a DFA $M_3$.
           \item  Construct a $\DPCM$ $M_4$ accepting $L(M_2) \cap L(M_3)$.
	\item Check if the language accepted by $M_4$ is empty.
                    This is possible since the emptiness problem for
		    $\NPCM$s (hence also for $\DPCM$s) is decidable
		    \cite{Ibarra1978}.
	\end{enumerate}

	\noindent
	By Lemma \ref{lemmaNA}, $a^d \in L(Z)$ if and only if $\$T\$ \not \subseteq
	\$ L_{na}\$$ if and only if the language
	accepted by $L(M_4) $ is not empty.
        It follows that $\varrho(L) \notin \DPCM$.

Similarly with $\$L_{na}\$$ for $\NCM(1,1)$ for suffix and infix, and
$\%\$L_{na}\$$ for outfix.

\
 \end{proof}
We note that the proof above also shows that if $L \in \NCM(1,1)$,
then $\pref(L)$ need not be in $\DPCM$.

\section{Bounded-Dense Languages}
\label{boundeddense}

Let $\varrho$ be an operation from $\Sigma^*$ to $\Sigma^*$.
Then a language $L$  is $\varrho$-bounded-dense over given words
$w_1,  \ldots , w_k $  if  $\varrho(L) =  w_1^* \cdots w_k^*$.
We will show below that determining bounded-denseness
is decidable for $\NPCM$ languages.

The following lemma is a generalization of
a similar result for $\NPDA$s in \cite{GinsburgCFLbook}:
\begin{lemma} \label{bounded1}
It is decidable, given two $\NPCM$s $M_1$ and $M_2$,
one of which accepts a bounded language that is a subset of
$w_1^* \cdots w_k^*$ (for given 
words $w_1, \ldots, w_k \in \Sigma^+$), whether $L(M_1) \subseteq L(M_2)$.
\end{lemma}
\begin{proof}
We consider two cases.


\noindent
Case 1:  Suppose $L(M_2) \subseteq w_1^* \cdots w_k^*$.
From a bounded $\NPCM$ language, it is known that 
we can  construct a $\DCM$ machine $M_2'$
equivalent to $M_2$ \cite{boundedSemilin}. Then, we can
also construct a $\DCM$ $M_2''$ accepting $\overline{L(M_2')}$
\cite{Chiniforooshan2012}.  Next, we construct an $\NPCM$ $M$ that
simulates $M_1$ and $M_2''$ in parallel to accept
$L(M_1) \cap L(M_2'')$.  Clearly, $L(M_1) \subseteq L(M_2)$
if and only if  $L(M) = \emptyset$, which is decidable since the
emptiness problem for $\NPCM$s is decidable \cite{Ibarra1978}.

\vskip .25cm

\noindent
Case 2: Suppose $L(M_1) \subseteq w_1^*  \cdots w_k^*$.
First we construct an $\NPCM$ $M_2'$ that accepts
$L(M_2) \cap w_1^* \cdots w_k^*$.   Then
$L(M_1) \subseteq L(M_2)$ if and only if
$L(M_1)  \subseteq L(M_2')$, which is decidable by Case 1.

\end{proof}

\begin{cor}
It is decidable, given two $\NPCM$s $M_1, M_2$ accepting bounded
languages $L(M_1),L(M_2) \subseteq w_1^* \cdots w_k^*$,  whether
$L(M_1) \subseteq L(M_2)$  (resp., $L(M_ 1) = L(M_2)$).
\end{cor}

Let 
$\varrho \in  \{\suff, \inf, \pref, \outf \}$.
Clearly, if $M$ is an $\NPCM$ accepting a language
$L(M)\subseteq w_1^* \cdots w_k^*$, we
can construct an $\NPCM$ $M'$ such that
$L(M') = \varrho (L(M))$, and $L(M')$ is bounded, but over 
$v_1^* \cdots v_l^*$, which are effectively constructable from
$w_1, \ldots, w_k$.
From Lemma \ref{bounded1}, by testing $w_1^* \cdots w_k^* \subseteq L(M')$ we have:
\begin{prop}
Let $\varrho \in \{\pref, \inf, \suff, \outf\}$.
It is decidable, given an $\NPCM$ $M$ accepting a
language $L(M) \subseteq w_1^* \cdots w_k^*$
(for given $w_1, \ldots, w_k$), whether
$L(M)$ is $\varrho$-bounded-dense. 
\end{prop}

\section{Conclusions}

This paper studies decidability problems involving testing whether
a language $L$ is $\varrho$-dense and $\varrho$-marked-dense, depending on
the language family of $L$. For the prefix operation, all are decidable for $\DCM$,
but undecidable for $\NCM(1,1)$, and thus the problem has been completely
characterized in terms of restrictions on reversal-bounded multicounter
machines. For suffix, both density and marked-density are decidable for
$\DCM(1,1)$, but not for $\DCM(1,3)$ and $\NCM(1,1)$, and therefore this
has also been completely characterized. For infix, marked-density is
decidable for $\DCM(1,1)$, but not for $\DCM(1,3)$ and $\NCM(1,1)$.
For infix-density however, it is decidable for nondeterministic reversal-bounded pushdown automata,
but undecidable for deterministic pushdown automata and nondeterministic
one-counter automata. It remains open for $\DCM$ and $\NCM$ when there are
at least two counters, and also for deterministic one-counter automata.
For outfix, marked-density is undecidable for $\DCM(1,3), \DCM(2,1)$ and
$\NCM(1,1)$ but is open for $\DCM(1,1)$. All variants are open for outfix-density.

In Section \ref{boundeddense}, results on bounded-dense languages are presented where the words $w_1, \ldots, w_k$ are given. 
In particular, for each of prefix, infix, suffix and outfix, it is decidable
for $\NPCM$ languages that accept bounded languages, whether they
are $\varrho$-bounded-dense.
%

\begin{table}\begin{center}
\begin{tabular}{|l|c|c|c|c|}
\hline
\textbf{unmarked density} & infix  & suffix & prefix & outfix
\\\hline 
 \hline $\DCM(1,1)$  & $\checkmark_{C3}$ & $\checkmark_{C3}$ & $\checkmark_{C2}$ & ?
 \\\hline $\DCM(1,3)$  & $\checkmark_{P10}$ & $\times_{P7}$ & $\checkmark_{C2}$ &?
 \\\hline $\DCM(2,1)$  & ? & $\times_{P7}$ & $\checkmark_{C2}$ & ?
 \\\hline $\DCM$     & ? & $\times_{P7}$ & $\checkmark_{C2}$ & ?
 \\\hline $\DPDA(1)$     & $\checkmark_{P10}$ & ? & $\checkmark_{P2}$ & ?
 \\\hline $\DPDA$     & $\times_{P9}$ & $\times_{P7}$ & $\checkmark_{P2}$ & ?
 \\\hline $\NCM(1,1)$  & $\checkmark_{P10}$ & $\times_{P7}$ & $\times_{C4}$& ?
 \\\hline $\NCM$     & ? & $\times_{P7}$ & $\times_{C4}$ & ?
 \\\hline $\NCA$     & $\times_{P8}$ & $\times_{P7}$ & $\times_{C4}$ & ?
 \\\hline rev-$\NPDA$     & $\checkmark_{P10}$ & $\times_{P7}$ & $\times_{C4}$ & ?
 \\\hline $\NPDA$     & $\times_{P8}$ & $\times_{P7}$ & $\times_{C4}$ & ?
\\ \hline \hline
\textbf{marked density} & infix  & suffix & prefix & outfix
\\\hline 
 \hline $\DCM(1,1)$  & $\checkmark_{C3}$ & $\checkmark_{C3}$ & $\checkmark_{C2}$ & ?
 \\\hline $\DCM(1,3)$  & $\times_{P4}$ & $\times_{P4}$ & $\checkmark_{C2}$ & $\times_{P5}$
 \\\hline $\DCM(2,1)$  & $\times_{P4}$ & $\times_{P4}$ & $\checkmark_{C2}$ & $\times_{P5}$
 \\\hline $\DCM$     & $\times_{P4}$ & $\times_{P4}$ & $\checkmark_{C2}$ & $\times_{P5}$
 \\\hline $\DPDA(1)$     & $\times_{P6}$ & $\times_{P6}$ & $\checkmark_{P2}$ & $\times_{P6}$
 \\\hline $\DPDA$     & $\times_{P6}$ & $\times_{P6}$ & $\checkmark_{P2}$ & $\times_{P6}$
 \\\hline $\NCM(1,1)$ & $\times_{P4}$ & $\times_{P4}$ & $\times_{P4}$ & $\times_{P5}$
 \\\hline $\NCM$     & $\times_{P4}$ & $\times_{P4}$ & $\times_{P4}$ & $\times_{P5}$
 \\\hline $\NCA$     & $\times_{P4}$ & $\times_{P4}$ & $\times_{P4}$ & $\times_{P5}$
 \\\hline rev-$\NPDA$   & $\times_{P4}$ & $\times_{P4}$ & $\times_{P4}$ & $\times_{P5}$
 \\\hline $\NPDA$  & $\times_{P4}$ & $\times_{P4}$ & $\times_{P4}$ & $\times_{P5}$
 \\\hline
\end{tabular} \end{center}
\caption{Summary of results, in the top half of the table with different types of density, and in the bottom half of the table with different types of marked density. A checkmark represents decidability, a cross is undecidable, and a question mark represents an open problem. The proposition proving each result is listed as subscript (with C being a corollary, and P a proposition).}
\label{tab:summary}
\end{table}


\end{document}